\documentclass[runningheads,a4paper]{llncs}

\usepackage{amssymb, amsmath}
\usepackage{graphicx}
\usepackage{cite}
\usepackage{url}
\usepackage{enumitem}
\usepackage{multirow}

\urldef{\mailsa}\path|manoj_raut@daiict.ac.in|
\newcommand{\keywords}[1]{\par\addvspace\baselineskip
\noindent\keywordname\enspace\ignorespaces#1}

\begin{document}

\mainmatter  

\title{An Algorithm for Computing Prime Implicates in Modal Logic using Resolution}

\titlerunning{An Algorithm for Computing Prime Implicates in Modal Logic}

%
%
\author{Manoj K. Raut}
%
\authorrunning{Manoj K. Raut}

\institute{Dhirubhai Ambani Institute of Information and Communication Technology,\\
Gandhinagar, Gujarat-382007, India\\ Phone: +91 79 3051 0585\\ Fax: +91 79 3052 0010\\
\mailsa\\
}

%
%

\toctitle{Lecture Notes in Computer Science}
\tocauthor{Authors' Instructions}
\maketitle

\begin{abstract}
In this paper we have proposed an algorithm for computing prime implicates of a modal formula in $\mathbf{K}$ using resolution method suggested in \cite{Enjalbert}. The algorithm suggested in this paper takes polynomial times exponential time ,i.e, $O(n^{2k}\times 2^{n})$ to compute prime implicates whereas Binevenu's algorithm \cite{Bienvenu} takes doubly exponential time to compute prime implicates.  We have also proved its correctness.

\end{abstract}
\\

\section{Introduction}
Knowledge representation and retrieval is a fundamental issue in artificial intelligence . An agent stores what it knows in a knowledge base $X$ using a logical formalism. Then queries are thrown at the knowledge base $X$ to extract implicit information stored in it. Mathematically, if we have a knowledge base $X$ expressed in some logical formalism and a query $Q$ in hand then the logical entailment problem is whether $X\models Q$. This problem is intractable \cite{Cook} in general as every known algorithm runs in time exponential time in the size of the given knowledge base. To overcome such computational intractability, the logical entailment problem is split into two phases such as off-line and on-line. In the off-line phase, the original knowledge base $X$ is preprocessed to obtain a new knowledge base $X^{'}$ and in on-line phase the output of the compilation, i.e, $X^{'}$ is used for query answering in polynomial time. In such type of compilation most of the computational overhead shifted into the off-line phase, is amortized over exponential number of queries in on-line query answering. The off-line computation is called {\em knowledge compilation}.

Various approaches of knowledge compilation in propositional logic, first order logic and modal logic has been proposeded so far in literature \cite{Bienvenu, Bienvenu1, Cadoli, Coudert1, Darwiche, Jackson1, Kean1, Kleer1, Kleer2, Ngair, Manoj, Manoj1, Manoj2, Reiter, Shiny, Slagle, Strzemecki, Tison}. One of the proposed approaches is to calculate  the prime implicates/implicants $\Pi(X)$ of a knowledge base $X$ in the off-line phase and queries are answered from $\Pi(X)$ in on-line phase in polynomial time. 

Most of the work of this kind of knowledge
compilation have been suggested in propositional logic, first order logic and in modal logic. Due to lack of expressive power in propositional logic and the undecidability of first order logic, modal logic is required as a knowledge representation language in many problems. Modal logic gives a trade-off between expressivity and complexity as they are more expressive than propositional logic and computationally better behaved than first order logic. An
algorithm to compute the set of prime implicates of modal logic $\mathbf{K}$ and $\mathbf{K}_{n}$ have been suggested in \cite{Bienvenu} and \cite{Bienvenu1} respectively using distributive method. An incremental algorithm for computing prime implicates in modal logic is suggested in \cite{Manoj1} using distributivity property and an algorithm for computing theory prime implicates in modal logic is suggested in \cite{Manoj2} using also distributivity property.  In this paper we propose to compute prime implicate from a modal knowledge base using resolution method suggested in \cite{Enjalbert, Areces}. The algorithm suggested in this paper is more efficient than the algorithm suggested in \cite{Bienvenu}.

The paper is organized as follows. In section 2 we give basic definitions and direct resolution method in modal logic. In Section 3 we describe the algorithm, its soundness, completeness and complexity for computing prime implicates in modal logic. Section 4 concludes the paper. 

\section{Preliminaries}

Let us now discuss the basics of modal logic $\mathbf{K}$ briefly from \cite{Blackburn1, Blackburn2}. The alphabet of modal formulas is $Var\cup\{\neg,\wedge,\vee, \Diamond,(,)\}$. $Var$ is a countable set of propositional letters denoted by $p,q,r,\ldots$. The connectives $\neg$, $\wedge$, and $\vee$ are negation, conjunction and disjunction. $\Diamond$ is the modal operator `possible'. The modal formulas $MF$ are defined inductively as follows. Propositional letters are modal formulas. If $A$ and $B$ are modal formulas then $\neg A, A\wedge B, A\vee B, \Diamond A$ are modal formulas. For the sake of convenience, we introduce the connectives $A\rightarrow B\equiv\neg A\vee B, A\leftrightarrow B\equiv(A\rightarrow B)\wedge(B\rightarrow A)$. The `necessary' operator $\Box$ is defined as $\Box A\equiv\neg\Diamond\neg A$. We avoid using parentheses whenever possible. The length of a formula $\phi$, denoted by $|\phi|$, is the number of occurrences of propositional variables, logical connectives, and modal operators in $\phi$. We now present the semantics of modal logic $\mathbf{K}$.

\begin{definition}\label{model}
A Kripke model $M$ is a triple $\langle W, R, V\rangle$ where $W$ is a nonempty set (of worlds), $R$ is a binary relation on $W$ called the accessibility relation , and $V: Var\rightarrow 2^{W}$ is a valuation function, which assigns to each proposition letter $p\in Var$ a subset of $W$. If $p\in Var$ then $V(p)$ is the set of states at which $p$ is true.
\end{definition}

\begin{definition}\label{model_truth}
Given any Kripke model $M=\langle W, R, V\rangle$, a world $w\in W$, and a formula $\phi\in MF$, the truth of $\phi$ at $w$ of $M$ denoted by $M,w\models\phi$, is defined inductively as follows:
\begin{itemize}
 \item $M,w\models p$ iff $w\in V(p)$ where $p\in Var$,
 \item $M,w\models\neg\phi$ iff $M,w\not\models\phi$, 
\item $M,w\models\phi\vee\psi$ iff $M,w\models\phi$ or $M,w\models\psi$,
 \item $M,w\models\phi\wedge\psi$ iff $M,w\models\phi$ and $M,w\models\psi$,
\item $M,w\models\Diamond\phi$ iff for some $w^{'}\in W$ with $wRw^{'}$ we have $M,w^{'}\models\phi$,
 \item $M,w\models\Box\phi$ iff for all $w^{'}\in W$ with $wRw^{'}$ we have $M,w^{'}\models\phi$. 
  \end{itemize}
\end{definition}

We say that a formula $\phi$ is satisfiable if there exists a model $M$ and and a world $w$ such that $M,w\models\phi$ and say $\phi$ is valid denoted by $\models\phi$ if $M,w\models\phi$ for all $M$ and $w$. A formula $\phi$ is unsatisfiable written as $\phi\models \bot$ if there exists no $M$ and $w$ for which $M,w\models\phi$. A formula $\psi$ is a logical consequence of a formula $\phi$ written as $\phi\models\psi$ if $M,w\models\phi$ implies $M,w\models\psi$ for every model $M$ and world $w\in W$.

There are two types of logical consequences given in \cite{Blackburn1} in modal logic which are:

\begin{enumerate}
\item a formula $\psi$ is a global consequence of $\phi$ if whenever $M,w\models\phi$ for every world $w$ of a model $M$, then $M,w\models\psi$ for every world $w$ of $M$.
\item a formula $\psi$ is a local consequence of $\phi$ if $M,w\models\phi$ implies $M,w\models\psi$ for every model $M$ and world $w$.
\end{enumerate}

Eventhough both consequences exist, in this paper we will only study local consequences and whenever $\phi\models\psi$ we mean $\psi$ is a local consequence of $\phi$.


A modal formula is in disjunctive normal form(DNF) if it is a (possibly empty) disjunction of the form
$$\eta=L_1\vee L_2\vee\ldots\vee L_u\vee\Box D_1\vee\Box D_2\vee\ldots\vee \Box D_v\vee\Diamond A_1\vee \Diamond A_2\vee\ldots\vee \Diamond A_w$$ 
where each $L_i$ is a literal, each $D_i$ is in disjunctive normal form, and each $A_i$ is in conjunctive normal form. A modal formula is in conjunctive normal form(CNF) if it is a conjunction $U=C_1\wedge C_2\wedge\ldots\wedge C_n$, where each $C_i$ is in disjunctive normal form. A formula in disjunctive normal form is called a clause. The empty clause is denoted by $\bot$. We identify the conjunction $C_1\wedge C_2\wedge\ldots\wedge C_n$ with the set $(C_1, C_2,\ldots,C_n)$. In this paper we shall always consider formulas in CNF because for any formula $X$ we can construct an equivalent formula $X^{'}$ in CNF in $\mathbf{K}$.

\begin{definition}
A clause $C$ is said to be an implicate of a formula $X$ if $X\models C$. A clause $C$ is a prime implicate of $X$ if $C$ is an implicate of $X$ and there is no other implicate $C^{'}$ of $X$ such that $C^{'}\models C$. The set of implicates and prime implicates of $X$ are denoted by $\Psi(X)$ and $\Pi(X)$ respectively. 
\end{definition}

\begin{definition}
Let $Y$ be a set of clauses. The residue
  of $Y$, denoted by ${Res}(Y)$ is a subset of $Y$ such
 that for every clause $C\in Y$, there is a clause $D\in {Res}(Y)$
 where $D\models C$; and no clause in ${Res}(Y)$ entails any
 other clause in ${Res}(Y)$.
\end{definition}

A direct resolution method introduced in Enjalbert and Farinas del Cerro \cite{Enjalbert} has been proved to be complete for $\mathbf{K}$. They have introduced the following resolution proof system in $\mathbf{K}$.

Define, by induction, two relations on clauses, (i) $C$ is a direct resolvent of $A$ and $B$, and (ii) $C$ is a direct resolvent of $A$, i.e, in symbol $\Sigma(A,B)\rightarrow C$ and $\Gamma(A)\rightarrow C$ respectively by the following formal system:

\newpage

\begin{table}[htbp]
\large
\caption{}
\begin{center}
\begin{tabular}{ |l|l| }
\hline
\multicolumn{2}{ |c| }{Axioms} \\
\hline
\multicolumn{2}{ |c| }{$(A1) \Sigma(p,\neg p)\rightarrow\bot$} \\
\multicolumn{2}{ |c| }{$(A1) \Sigma(\bot, A)\rightarrow\bot$} \\
\hline
$\Sigma$-rules & $\Gamma$-rules\\
\hline
\multirow{4}{*}{}  ($\vee$)$\frac{\Sigma(A,B)\rightarrow C}{\Sigma(A\vee D_1, B\vee D_2)\rightarrow C\vee D_1\vee D_2}$ & ($\Diamond 1$)$\frac{\Sigma(A,B)\rightarrow C}{\Gamma(\Diamond(A,B,F))\rightarrow\Diamond(A,B,C,F)}$ \\
  ($\Box\Diamond$)$\frac{\Sigma(A,B)\rightarrow C}{\Sigma(\Box A,\Diamond(B,E))\rightarrow\Diamond(B,C,E)}$ & ($\Diamond 2$)$\frac{\Gamma(A)\rightarrow B}{\Gamma(\Diamond(A,F))\rightarrow\Diamond(B,A,F)}$ \\
  ($\Box\Box$)$\frac{\Sigma(A,B)\rightarrow C}{\Sigma(\Box A,\Box B)\rightarrow\Box C}$ & ($\vee$) $\frac{\Gamma(A)\rightarrow B}{\Gamma(A\vee C)\rightarrow B\vee C}$ \\
   & ($\Box$)$\frac{\Gamma(A)\rightarrow B}{\Gamma(\Box A)\rightarrow\Box B}$ \\ \hline
\end{tabular}
\label{tab1}
\end{center}
\end{table}

\vspace {.2cm}

{\text{Fig 1:Enjalbert and Farinas del Cerro resolution rules.}}{$\hspace{5cm}$}

\vspace {.5cm}

\noindent where $A,B,C,D,D_1,D_2$ are clauses and $E, F$ are sets of clauses.  

Define the simplification relation $A\approx B$ as the least congruence relation containing:
\begin{itemize}
\item $\Diamond\bot\approx\bot$,
\item $\bot\vee D\approx D$,
\item $(\bot, E)\approx\bot$,
\item $A\vee A\vee D\approx A\vee D$.
\end{itemize}

For any formula $X$ there is a unique formula $X^{'}$ such that $X\approx X^{'}$ and $X^{'}$ cannot be simplified further. The formula $X^{'}$ is called the normal form of $X$. $C$ is a resolvent of $A$ and $B$ (respectively, of $A$) iff there is some $C^{'}$ such that $\Sigma(A,B)\rightarrow C^{'}$ (respectively, $\Gamma(A)\rightarrow C^{'}$) and $C$ is the normal form of $C^{'}$. We write  $\Sigma(A,B)\Rightarrow C$ (respectively, $\Gamma(A)\Rightarrow C$) if $C$ is a resolvent of $A$ and $B$ (respectively, of $A$).

\section{A Direct Resolution Method for Computing Prime Implicates}

The following lemma says that the formula is unsatisfiable iff the following seven conditions hold. This lemma is used in the proof of Theorem \ref{several_possibilities_wrt_closure}. 

\begin{lemma}\label{several_possibilities}
Let $\beta_1,\beta_2,\ldots,\beta_m,\gamma_1,\gamma_2,\ldots,\gamma_n, \phi_1,\phi_2,\ldots,\phi_q,\xi_1,\xi_2,\ldots,\xi_r$ be modal formulas and $\alpha_1,\alpha_2,\ldots,\alpha_l, \psi_1,\psi_2,\ldots,\psi_p$ be propositional formulas. Then
$[(\wedge_{i=1}^{l}\alpha_i)\vee(\wedge_{j=1}^{m}\Diamond\beta_j)\vee(\wedge_{k=1}^{n}\Box\gamma_{k})\vee((\wedge_{i=1}^{l}\alpha_i)\wedge(\wedge_{j=1}^{m}\Diamond\beta_j))\vee((\wedge_{i=1}^{l}\alpha_i)\wedge(\wedge_{k=1}^{n}\Box\gamma_{k}))\vee((\wedge_{j=1}^{m}\Diamond\beta_j)\wedge(\wedge_{k=1}^{n}\Box\gamma_{k}))\vee((\wedge_{i=1}^{l}\alpha_i)\wedge(\wedge_{j=1}^{m}\Diamond\beta_j)\wedge(\wedge_{k=1}^{n}\Box\gamma_{k}))]\wedge\psi_1\wedge\ldots\wedge\psi_p\wedge\Box\phi_1\wedge\ldots\wedge\Box\phi_q\wedge\Diamond\xi_1\wedge\ldots\wedge\Diamond\xi_r\models\bot$ if and only if
\begin{enumerate}
\item $(\wedge_{i=1}^{l}\alpha_i)\wedge\psi_1\wedge\ldots\wedge\psi_p\models\bot$ or
\item $(\wedge_{j=1}^{m}\beta_j)\wedge\phi_1\wedge\ldots\wedge\phi_q\models\bot$ or
\item $(\wedge_{k=1}^{n}\gamma_{k})\wedge\xi_u\wedge\phi_1\wedge\ldots\wedge\phi_q\models\bot$ for $1\leq u\leq r$ or
\item $((\wedge_{i=1}^{l}\alpha_i)\wedge(\wedge_{j=1}^{m}\beta_j))\wedge\phi_1\wedge\ldots\wedge\phi_q\models\bot$ or
\item $((\wedge_{i=1}^{l}\alpha_i)\wedge(\wedge_{k=1}^{n}\gamma_{k}))\wedge\xi_u\wedge\phi_1\wedge\ldots\wedge\phi_q\models\bot$ for $1\leq u\leq r$ or
\item $((\wedge_{j=1}^{m}\beta_j)\wedge(\wedge_{k=1}^{n}\gamma_{k}))\wedge\phi_1\wedge\ldots\wedge\phi_q\models\bot$ or
\item $((\wedge_{i=1}^{l}\alpha_i)\wedge(\wedge_{j=1}^{m}\beta_j)\wedge(\wedge_{k=1}^{n}\gamma_{k}))\wedge\phi_1\wedge\ldots\wedge\phi_q\models\bot$.
\end{enumerate}
\end{lemma}
\begin{proof}
Suppose 
\begin{enumerate}
\item $(\wedge_{i=1}^{l}\alpha_i)\wedge\psi_1\wedge\ldots\wedge\psi_p\not\models\bot$,  
 \item $(\wedge_{j=1}^{m}\beta_j)\wedge\phi_1\wedge\ldots\wedge\phi_q\not\models\bot$, 
 \item $(\wedge_{k=1}^{n}\gamma_{k})\wedge\xi_u\wedge\phi_1\wedge\ldots\wedge\phi_q\not\models\bot$ for $1\leq u\leq r$,
\item $((\wedge_{i=1}^{l}\alpha_i)\wedge(\wedge_{j=1}^{m}\beta_j))\wedge\phi_1\wedge\ldots\wedge\phi_q\not\models\bot$,
 \item $((\wedge_{i=1}^{l}\alpha_i)\wedge(\wedge_{k=1}^{n}\gamma_{k}))\wedge\xi_u\wedge\phi_1\wedge\ldots\wedge\phi_q\not\models\bot$ for $1\leq u\leq r$,
 \item $((\wedge_{j=1}^{m}\beta_j)\wedge(\wedge_{k=1}^{n}\gamma_{k}))\wedge\phi_1\wedge\ldots\wedge\phi_q\not\models\bot$, and
\item $((\wedge_{i=1}^{l}\alpha_i)\wedge(\wedge_{j=1}^{m}\beta_j)\wedge(\wedge_{k=1}^{n}\gamma_{k}))\wedge\phi_1\wedge\ldots\wedge\phi_q\not\models\bot$. 
\end{enumerate}

Let there be a propositional model $w$ of $(\wedge_{i=1}^{l}\alpha_i)\wedge\psi_1\wedge\ldots\wedge\psi_p$ and for others we have a model $\mathcal{M}^{'}$ and a state $w^{'}$ such that
\begin{enumerate} 
\item $\mathcal{M}, w\models(\wedge_{i=1}^{l}\alpha_i)\wedge\psi_1\wedge\ldots\wedge\psi_p$
\item $\mathcal{M}^{'},w^{'}\models(\wedge_{j=1}^{m}\beta_j)\wedge\phi_1\wedge\ldots\wedge\phi_q$, 
 \item $\mathcal{M}^{'},w^{'}\models(\wedge_{k=1}^{n}\gamma_{k})\wedge\xi_u\wedge\phi_1\wedge\ldots\wedge\phi_q$ for $1\leq u\leq r$,
 \item $\mathcal{M}^{'},w^{'}\models((\wedge_{i=1}^{l}\alpha_i)\wedge(\wedge_{j=1}^{m}\beta_j))\wedge\phi_1\wedge\ldots\wedge\phi_q$,
 \item $\mathcal{M}^{'},w^{'}\models((\wedge_{i=1}^{l}\alpha_i)\wedge(\wedge_{k=1}^{n}\gamma_{k}))\wedge\xi_u\wedge\phi_1\wedge\ldots\wedge\phi_q$ for $1\leq u\leq r$,
 \item $\mathcal{M}^{'},w^{'}\models((\wedge_{j=1}^{m}\beta_j)\wedge(\wedge_{k=1}^{n}\gamma_{k}))\wedge\phi_1\wedge\ldots\wedge\phi_q$, and
\item $\mathcal{M}^{'},w^{'}\models((\wedge_{i=1}^{l}\alpha_i)\wedge(\wedge_{j=1}^{m}\beta_j)\wedge(\wedge_{k=1}^{n}\gamma_{k}))\wedge\phi_1\wedge\ldots\wedge\phi_q$.
\end{enumerate}

Now we construct a new model $\mathcal{M}$ which contains the model $\mathcal{M}^{'}$, $w^{'}$ and  a relation $Rww^{'}$ for each $w^{'}$ then the above statements become 
\begin{enumerate}
\item $\mathcal{M},w\models(\wedge_{i=1}^{l}\alpha_i)$, $\mathcal{M},w\models\psi_1\wedge\ldots\wedge\psi_p$, 
\item $\mathcal{M},w\models(\wedge_{j=1}^{m}\Diamond\beta_j)$, $\mathcal{M},w\models\Box\phi_1\wedge\ldots\wedge\Box\phi_q$,
\item $\mathcal{M},w\models(\wedge_{k=1}^{n}\Box\gamma_{k})$, $\mathcal{M},w\models\Diamond\xi_u\wedge\Box\phi_1\wedge\ldots\wedge\Box\phi_q$ for $1\leq u\leq r$,
 \item $\mathcal{M}^{'},w^{'}\models(\wedge_{i=1}^{l}\alpha_i)$, $\mathcal{M}^{'},w^{'}\models(\wedge_{j=1}^{m}\beta_j)$, $\mathcal{M}^{'},w^{'}\models\phi_1\wedge\ldots\wedge\phi_q$, Considering a propositional model $w$ of $(\wedge_{i=1}^{l}\alpha_i)$ we get $(\mathcal{M},w\models(\wedge_{i=1}^{l}\alpha_i)$. Construct a new model $\mathcal{M}$ which contains the model $\mathcal{M}^{'}$, $w^{'}$ and a relation $Rww^{'}$ for each $w^{'}$ then 
$\mathcal{M},w\models(\wedge_{j=1}^{m}\Diamond\beta_j)$, $\mathcal{M},w\models\Box\phi_1\wedge\ldots\wedge\Box\phi_q$. So  $\mathcal{M},w\models(\wedge_{i=1}^{l}\alpha_i)\wedge(\wedge_{j=1}^{m}\Diamond\beta_j)$, $\mathcal{M},w\models\Box\phi_1\wedge\ldots\wedge\Box\phi_q$. Similarly,   
 \item $\mathcal{M},w\models((\wedge_{i=1}^{l}\alpha_i)\wedge(\wedge_{k=1}^{n}\Box\gamma_{k}))$, $\mathcal{M},w\models\Diamond\xi_u\wedge\Box\phi_1\wedge\ldots\wedge\Box\phi_q$ for $1\leq u\leq r$,
 \item $\mathcal{M},w\models((\wedge_{j=1}^{m}\Diamond\beta_j)\wedge(\wedge_{k=1}^{n}\Box\gamma_{k}))$, $\mathcal{M},w\models\Box\phi_1\wedge\ldots\wedge\Box\phi_q$, and
\item $\mathcal{M},w\models((\wedge_{i=1}^{l}\alpha_i)\wedge(\wedge_{j=1}^{m}\Diamond\beta_j)\wedge(\wedge_{k=1}^{n}\Box\gamma_{k}))$, $\mathcal{M},w\models\Box\phi_1\wedge\ldots\wedge\Box\phi_q$. 
\end{enumerate}
So, $\mathcal{M}, w\models  [(\wedge_{i=1}^{l}\alpha_i)\vee(\wedge_{j=1}^{m}\Diamond\beta_j)\vee(\wedge_{k=1}^{n}\Box\gamma_{k})\vee((\wedge_{i=1}^{l}\alpha_i)\wedge(\wedge_{j=1}^{m}\Diamond\beta_j))\vee((\wedge_{i=1}^{l}\alpha_i)\wedge(\wedge_{k=1}^{n}\Box\gamma_{k}))\vee((\wedge_{j=1}^{m}\Diamond\beta_j)\wedge(\wedge_{k=1}^{n}\Box\gamma_{k}))\vee((\wedge_{i=1}^{l}\alpha_i)\wedge(\wedge_{j=1}^{m}\Diamond\beta_j)\wedge(\wedge_{k=1}^{n}\Box\gamma_{k}))]\wedge\psi_1\wedge\ldots\wedge\psi_p\wedge\Box\phi_1\wedge\ldots\wedge\Box\phi_q\wedge\Diamond\xi_1\wedge\ldots\wedge\Diamond\xi_r$. This implies, $[(\wedge_{i=1}^{l}\alpha_i)\vee(\wedge_{j=1}^{m}\Diamond\beta_j)\vee(\wedge_{k=1}^{n}\Box\gamma_{k})\vee((\wedge_{i=1}^{l}\alpha_i)\wedge(\wedge_{j=1}^{m}\Diamond\beta_j))\vee((\wedge_{i=1}^{l}\alpha_i)\wedge(\wedge_{k=1}^{n}\Box\gamma_{k}))\vee((\wedge_{j=1}^{m}\Diamond\beta_j)\wedge(\wedge_{k=1}^{n}\Box\gamma_{k}))\vee((\wedge_{i=1}^{l}\alpha_i)\wedge(\wedge_{j=1}^{m}\Diamond\beta_j)\wedge(\wedge_{k=1}^{n}\Box\gamma_{k}))]\wedge\psi_1\wedge\ldots\wedge\psi_p\wedge\Box\phi_1\wedge\ldots\wedge\Box\phi_q\wedge\Diamond\xi_1\wedge\ldots\wedge\Diamond\xi_r\not\models\bot$

Conversely, suppose $[(\wedge_{i=1}^{l}\alpha_i)\vee(\wedge_{j=1}^{m}\Diamond\beta_j)\vee(\wedge_{k=1}^{n}\Box\gamma_{k})\vee((\wedge_{i=1}^{l}\alpha_i)\wedge(\wedge_{j=1}^{m}\Diamond\beta_j))\vee((\wedge_{i=1}^{l}\alpha_i)\wedge(\wedge_{k=1}^{n}\Box\gamma_{k}))\vee((\wedge_{j=1}^{m}\Diamond\beta_j)\wedge(\wedge_{k=1}^{n}\Box\gamma_{k}))\vee((\wedge_{i=1}^{l}\alpha_i)\wedge(\wedge_{j=1}^{m}\Diamond\beta_j)\wedge(\wedge_{k=1}^{n}\Box\gamma_{k}))]\wedge\psi_1\wedge\ldots\wedge\psi_p\wedge\Box\phi_1\wedge\ldots\wedge\Box\phi_q\wedge\Diamond\xi_1\wedge\ldots\wedge\Diamond\xi_r\not\models\bot$. Then there exists $\mathcal{M}$ and $w$ such that $\mathcal{M}, w\models[(\wedge_{i=1}^{l}\alpha_i)\vee(\wedge_{j=1}^{m}\Diamond\beta_j)\vee(\wedge_{k=1}^{n}\Box\gamma_{k})\vee((\wedge_{i=1}^{l}\alpha_i)\wedge(\wedge_{j=1}^{m}\Diamond\beta_j))\vee((\wedge_{i=1}^{l}\alpha_i)\wedge(\wedge_{k=1}^{n}\Box\gamma_{k}))\vee((\wedge_{j=1}^{m}\Diamond\beta_j)\wedge(\wedge_{k=1}^{n}\Box\gamma_{k}))\vee((\wedge_{i=1}^{l}\alpha_i)\wedge(\wedge_{j=1}^{m}\Diamond\beta_j)\wedge(\wedge_{k=1}^{n}\Box\gamma_{k}))]\wedge\psi_1\wedge\ldots\wedge\psi_p\wedge\Box\phi_1\wedge\ldots\wedge\Box\phi_q\wedge\Diamond\xi_1\wedge\ldots\wedge\Diamond\xi_r$. This implies the following must hold. 
\begin{enumerate}
\item $\mathcal{M},w\models(\wedge_{i=1}^{l}\alpha_i)\wedge\psi_1\wedge\ldots\wedge\psi_p\wedge\Box\phi_1\wedge\ldots\wedge\Box\phi_q\wedge\Diamond\xi_1\wedge\ldots\wedge\Diamond\xi_r$, 
\item $\mathcal{M},w\models(\wedge_{j=1}^{m}\Diamond\beta_j)\wedge\psi_1\wedge\ldots\wedge\psi_p\wedge\Box\phi_1\wedge\ldots\wedge\Box\phi_q\wedge\Diamond\xi_1\wedge\ldots\wedge\Diamond\xi_r$,
\item $\mathcal{M},w\models(\wedge_{k=1}^{n}\Box\gamma_{k})\wedge\psi_1\wedge\ldots\wedge\psi_p\wedge\Box\phi_1\wedge\ldots\wedge\Box\phi_q\wedge\Diamond\xi_1\wedge\ldots\wedge\Diamond\xi_r$,
\item $\mathcal{M},w\models((\wedge_{i=1}^{l}\alpha_i)\wedge(\wedge_{j=1}^{m}\Diamond\beta_j))\wedge\psi_1\wedge\ldots\wedge\psi_p\wedge\Box\phi_1\wedge\ldots\wedge\Box\phi_q\wedge\Diamond\xi_1\wedge\ldots\wedge\Diamond\xi_r$,
 \item $\mathcal{M},w\models((\wedge_{i=1}^{l}\alpha_i)\wedge(\wedge_{k=1}^{n}\Box\gamma_{k}))\wedge\psi_1\wedge\ldots\wedge\psi_p\wedge\Box\phi_1\wedge\ldots\wedge\Box\phi_q\wedge\Diamond\xi_1\wedge\ldots\wedge\Diamond\xi_r$,
 \item $\mathcal{M},w\models((\wedge_{j=1}^{m}\Diamond\beta_j)\wedge(\wedge_{k=1}^{n}\Box\gamma_{k}))\wedge\psi_1\wedge\ldots\wedge\psi_p\wedge\Box\phi_1\wedge\ldots\wedge\Box\phi_q\wedge\Diamond\xi_1\wedge\ldots\wedge\Diamond\xi_r$, and
\item $\mathcal{M},w\models((\wedge_{i=1}^{l}\alpha_i)\wedge(\wedge_{j=1}^{m}\Diamond\beta_j)\wedge(\wedge_{k=1}^{n}\Box\gamma_{k}))\wedge\psi_1\wedge\ldots\wedge\psi_p\wedge\Box\phi_1\wedge\ldots\wedge\Box\phi_q\wedge\Diamond\xi_1\wedge\ldots\wedge\Diamond\xi_r$.
\end{enumerate}

This implies, 
\begin{enumerate}
\item $\mathcal{M},w\models(\wedge_{i=1}^{l}\alpha_i)\wedge\psi_1\wedge\ldots\wedge\psi_p$, 
\item $\mathcal{M},w\models(\wedge_{j=1}^{m}\Diamond\beta_j)\wedge\Box\phi_1\wedge\ldots\wedge\Box\phi_q$,
\item $\mathcal{M},w\models(\wedge_{k=1}^{n}\Box\gamma_{k})\wedge\Diamond\xi_u\wedge\Box\phi_1\wedge\ldots\wedge\Box\phi_q$ for $1\leq u\leq r$,
\item $\mathcal{M},w\models((\wedge_{i=1}^{l}\alpha_i)\wedge(\wedge_{j=1}^{m}\Diamond\beta_j))\wedge\Box\phi_1\wedge\ldots\wedge\Box\phi_q$,
 \item $\mathcal{M},w\models((\wedge_{i=1}^{l}\alpha_i)\wedge(\wedge_{k=1}^{n}\Box\gamma_{k}))\wedge\Diamond\xi_u\wedge\Box\phi_1\wedge\ldots\wedge\Box\phi_q$ for $1\leq u\leq r$,
 \item $\mathcal{M},w\models((\wedge_{j=1}^{m}\Diamond\beta_j)\wedge(\wedge_{k=1}^{n}\Box\gamma_{k}))\wedge\Box\phi_1\wedge\ldots\wedge\Box\phi_q$, and
\item $\mathcal{M},w\models((\wedge_{i=1}^{l}\alpha_i)\wedge(\wedge_{j=1}^{m}\Diamond\beta_j)\wedge(\wedge_{k=1}^{n}\Box\gamma_{k}))\wedge\Box\phi_1\wedge\ldots\wedge\Box\phi_q$.
\end{enumerate}

From (1), $(\wedge_{i=1}^{l}\alpha_i)\wedge\psi_1\wedge\ldots\wedge\psi_p$ is satisfiable, so $(\wedge_{i=1}^{l}\alpha_i)\wedge\psi_1\wedge\ldots\wedge\psi_p\not\models\bot$ . From (2), $(\wedge_{j=1}^{m}\beta_j)\wedge\phi_1\wedge\ldots\wedge\phi_q$ is satisfiable because for all $w^{'}$ such that $Rww^{'}$ and $\mathcal{M},w^{'}\models(\wedge_{j=1}^{m}\beta_j)\wedge\phi_1\wedge\ldots\wedge\phi_q$. So $(\wedge_{j=1}^{m}\beta_j)\wedge\phi_1\wedge\ldots\wedge\phi_q\not\models\bot$.  From (3), $(\wedge_{k=1}^{n}\gamma_{k})\wedge\xi_u\wedge\phi_1\wedge\ldots\wedge\phi_q$ for $1\leq u\leq r$ is satisfiable because for all $w^{'}$ such that $Rww^{'}$ and $\mathcal{M}, w^{'}\models(\wedge_{k=1}^{n}\gamma_{k})\wedge\xi_u\wedge\phi_1\wedge\ldots\wedge\phi_q$ for $1\leq u\leq r$. So $(\wedge_{k=1}^{n}\gamma_{k})\wedge\xi_u\wedge\phi_1\wedge\ldots\wedge\phi_q\not\models\bot$ for $1\leq u\leq r$. From (4), we have $\mathcal{M},w\models(\wedge_{i=1}^{l}\alpha_i)$ and $\mathcal{M},w\models(\wedge_{j=1}^{m}\Diamond\beta_j)\wedge\Box\phi_1\wedge\ldots\wedge\Box\phi_q$. As  $\mathcal{M},w\models(\wedge_{j=1}^{m}\Diamond\beta_j)\wedge\Box\phi_1\wedge\ldots\wedge\Box\phi_q$ so for all $w^{'}$ such that $Rww^{'}$ and $\mathcal{M},w^{'}\models(\wedge_{j=1}^{m}\beta_j)\wedge\phi_1\wedge\ldots\wedge\phi_q$. So $(\wedge_{j=1}^{m}\beta_j)\wedge\phi_1\wedge\ldots\wedge\phi_q$ is satisfiable. As $\mathcal{M},w\models(\vee_{i=1}^{l}\alpha_i)$, then there exists a propositional model $w^{'}$ of $\wedge_{i=1}^{l}\alpha_i$ such that $\mathcal{M},w^{'}\models\wedge_{i=1}^{l}\alpha_i$. So $(\wedge_{i=1}^{l}\alpha_i)$ is satisfiable. So $((\wedge_{i=1}^{l}\alpha_i)\wedge(\wedge_{j=1}^{m}\beta_j))\wedge\phi_1\wedge\ldots\wedge\phi_q$ is satisfiable. Hence $((\wedge_{i=1}^{l}\alpha_i)\wedge(\wedge_{j=1}^{m}\beta_j))\wedge\phi_1\wedge\ldots\wedge\phi_q\not\models\bot$. From (5), the proof of  $((\wedge_{i=1}^{l}\alpha_i)\wedge(\wedge_{k=1}^{n}\gamma_{k}))\wedge\xi_u\wedge\phi_1\wedge\ldots\wedge\phi_q\not\models\bot$ for $1\leq u\leq r$ is similar to (4). From (6), the proof of $((\wedge_{j=1}^{m}\beta_j)\wedge(\wedge_{k=1}^{n}\gamma_{k}))\wedge\phi_1\wedge\ldots\wedge\phi_q\not\models\bot$ can be done similarly like (2). From (7), the proof of $((\wedge_{i=1}^{l}\alpha_i)\wedge(\wedge_{j=1}^{m}\beta_j)\wedge(\wedge_{k=1}^{n}\gamma_{k}))\wedge\phi_1\wedge\ldots\wedge\phi_q\not\models\bot$ can be done similarly like (4). \hfill{$\Box$}
\end{proof}

The following theorem says that the direct resolution of two implicates produces an implicate.. While proving the following theorem we have taken into consideration both $\Sigma$-rules and $\Gamma$-rules.

\begin{theorem}\label{resolution_consequence}
Resolution of two implicates (or one implicate) of a formula is an implicate of the formula.
\end{theorem}
\begin{proof}
We will prove by induction on the resolution on clauses for each of $\Sigma$-rules and $\Gamma$-rules. Let $X$ be any arbitrary formula.

\vspace{.2cm}
 
\noindent{\it Basis}:~(i) Let $p$ and $\neg p$ be implicates of $X$. So, $X\models p$ and $X\models\neg p$. Hence $X\models p\wedge\neg p$, i.e, $X\models\bot$. So $\bot$ is an implicate of $X$. (ii) Let $\bot$ and $A$ be implicates of $X$. So, $X\models \bot$ and $X\models A$. Hence $X\models \bot\wedge A$, i.e, $X\models\bot$. So $\bot$ is an implicate of $X$. 

\vspace{.2cm}

\noindent{\it Induction}:~ Assume that $C$ is a direct resolvent of $A$ and $B$ by $\Sigma$-rule or $C$ is a direct resolvent of $A$ by $\Gamma$-rule, i.e, by $\Sigma$-rule if $A$ and $B$ are implicates of $X$ then $C$ is an implicate of $X$ and by $\Gamma$-rule if $A$ is an implicate of $X$ then $C$ is an implicate of $X$.

\vspace{.2cm}

\noindent{\it $\vee$-rule}: Let $A\vee D_1$ and $B\vee D_2$ are implicates of $X$. Then $X\models A\vee D_1$ and $X\models B\vee D_2$. Then ($X\models A$ or $X\models D_1$) and ($X\models B$ or $X\models D_2$). So by distributivity, ($X\models A$ and $X\models B$) or ($X\models A$ and $X\models D_2$) or ($X\models D_1$ and $X\models B$) or ($X\models D_1$ and $X\models D_2$). If $X\models A$ and $X\models B$ then $X\models C$ by induction. This implies $X\models C\vee D_1\vee D_2$. If $X\models A$ and $X\models D_2$ then $X\models C\vee D_1\vee D_2$. If $X\models D_1$ and $X\models B$ then $X\models C\vee D_1\vee D_2$. If $X\models D_1$ and $X\models D_2$ then $X\models C\vee D_1\vee D_2$. So $C\vee D_1\vee D_2$ is an implicate of $X$. 

\vspace{.2cm}

\noindent{\it $\Box\Diamond$-rule}: Let $\Box A$ and $\Diamond(B,E)$ be implicates of $X$. Then $X\models\Box A$ and $X\models\Diamond(B,E)$, i.e, $X\models\Box A\wedge\Diamond(B,E)$. Then there is a model $\mathcal{M}$ and a state $w$ such that if $\mathcal{M},w\models X$ then $\mathcal{M},w\models\Box A\wedge\Diamond(B,E)$, i.e, for all $w^{'}$ such that $Rww^{'}$, $\mathcal{M}, w^{'}\models A\wedge B\wedge E$, i.e, by induction $\mathcal{M}, w^{'}\models B\wedge C\wedge E$, i.e, $\mathcal{M},w\models\Diamond(B\wedge C\wedge E)$. This implies $X\models \Diamond(B\wedge C\wedge E)$. So $\Diamond(B\wedge C\wedge E)$ is an implicate of $X$.

\vspace{.2cm}

\noindent{\it $\Box\Box$-rule}: Let $\Box A$ and $\Box B$ be implicates of $X$. So $X\models\Box A$ and $X\models\Box B$, i.e, $X\models\Box A\wedge\Box B$. Then there is a model $\mathcal{M}$ and a state $w$ such that if $\mathcal{M},w\models X$ then $\mathcal{M},w\models\Box A\wedge\Box B$, i.e, for all states $w^{'}$ such that $Rww^{'}$ we have $\mathcal{M},w^{'}\models A\wedge B$, i.e, by induction $\mathcal{M},w^{'}\models C$, i.e, $\mathcal{M},w\models \Box C$, i.e, $X\models\Box C$. So $\Box C$ is an implicate of $X$.

\vspace{.2cm}

\noindent{\it $\Diamond$-rule 1}: Let $\Diamond(A,B,F)$ be an implicate of $X$. So $X\models\Diamond(A,B,F)$. Then there is a model $\mathcal{M}$ and a state $w$ such that if $\mathcal{M},w\models X$ then $\mathcal{M},w\models\Diamond(A,B,F)$,i.e, for some $w^{'}$ such that $Rww^{'}$ we have $\mathcal{M}, w^{'}\models A\wedge B\wedge F$, i.e, by induction $\mathcal{M}, w^{'}\models A\wedge B\wedge C\wedge F$, i.e, $\mathcal{M}, w\models\Diamond(A,B,C,F)$, i.e, $X\models\Diamond(A,B,C,F)$. So $\Diamond(A,B,C,F)$ is an implicate of $X$.

\vspace{.2cm}

\noindent{\it $\Diamond$-rule 2}: Let $\Diamond(A,F)$ be an implicate of $X$. So $X\models\Diamond(A,F)$. Then there is a model $\mathcal{M}$ and a state $w$ such that if $\mathcal{M},w\models X$ then $\mathcal{M},w\models\Diamond(A,F)$, i.e, for some $w^{'}$ such that $Rww^{'}$ we have $\mathcal{M}, w^{'}\models A\wedge F$, i.e, by induction $\mathcal{M}, w^{'}\models B\wedge A\wedge F$, i.e, $X\models\Diamond(B\wedge A\wedge F)$. So  $\Diamond(B\wedge A\wedge F)$ is an implicate of $X$.

\vspace{.2cm}

\noindent{\it $\vee$-rule}: Let $A\vee C$ be an implicate of $X$. So $X\models A\vee C$. Then there is a model $\mathcal{M}$ and a state $w$ such that if $\mathcal{M},w\models X$ then $\mathcal{M},w\models A\vee C$, i.e, by induction $\mathcal{M},w\models B\vee C$, i.e, $X\models B\vee C$. So $B\vee C$ is an implicate of $X$.

\vspace{.2cm}

\noindent{\it $\Box$-rule}: Let $\Box A$ be an implicate of $X$. So $X\models\Box A$. Then there is a model $\mathcal{M}$ and a state $w$ such that if $\mathcal{M},w\models X$ then $\mathcal{M},w\models\Box A$, i.e, for all $w^{'}$ such that $Rww^{'}$ we have $\mathcal{M}, w^{'}\models A$, i.e, by induction $\mathcal{M}, w^{'}\models B$, i.e, $\mathcal{M}, w\models \Box B$, i.e, $X\models\Box B$. So $\Box B$ is an implicate of $X$. \hfill{$\Box$}
\end{proof}

Let us discuss the computational aspects of prime implicates. For a set of clauses $U$, let ${L}(U)=U\cup \{ C : \Sigma(A,B)\Rightarrow C \text{~or~}  \Gamma(A)\Rightarrow C \text{~for any clauses~} A \text{~and~}\\ B \text{~in~} U (\text{~respectively, any clause~} A \text{~of~} U)\}$. We construct the sequence $U, {L}(U),\\ {L}({L}(U)),
\ldots$, i.e, ${L^{0}}(U)=U$, ${L^{n+1}}(U)={L}({L^{n}}(U))$ for
$n\geq 0$.  Define the {\it resolution closure} of $U$ as
$\overline{L}(U)= \cup_i\{{L^{i}}(U):i\in \mathbb N\}$.

\begin{theorem}\label{several_possibilities_wrt_closure}
For every implicate $C$ of $U$ there exists a clause $D\in\overline{L}(U)$ such that $D\models C$.
\end{theorem}

\begin{proof}
Let $U=C_1\wedge C_2\wedge\ldots\wedge C_n=(\alpha_{11}\vee\ldots\vee\alpha_{1u}\vee\Diamond\beta_{11}\vee\ldots\vee\Diamond\beta_{1v}\vee\Box\gamma_{11}\vee\ldots\vee\Box\gamma_{1w})\wedge (\alpha_{21}\vee\ldots\vee\alpha_{2u}\vee\Diamond\beta_{21}\vee\ldots\vee\Diamond\beta_{2v}\vee\Box\gamma_{21}\vee\ldots\vee\Box\gamma_{2w})\wedge\ldots\wedge(\alpha_{n1}\vee\ldots\vee\alpha_{nu}\vee\Diamond\beta_{n1}\vee\ldots\vee\Diamond\beta_{nv}\vee\Box\gamma_{n1}\vee\ldots\vee\Box\gamma_{nw})$ be a formula. Let $C=\psi_1\vee\ldots\vee\psi_p\vee\Diamond\phi_1\vee\ldots\vee\Diamond\phi_q\vee\Box\xi_1\vee\ldots\vee\Box\xi_r$ be a prime implicate of $U$. As $C$ is an implicate of $U$ so $U\models C$, i.e, $U\wedge\neg C\models\bot$. This implies  $(\alpha_{11}\vee\ldots\vee\alpha_{1u}\vee\Diamond\beta_{11}\vee\ldots\vee\Diamond\beta_{1v}\vee\Box\gamma_{11}\vee\ldots\vee\Box\gamma_{1w})\wedge (\alpha_{21}\vee\ldots\vee\alpha_{2u}\vee\Diamond\beta_{21}\vee\ldots\vee\Diamond\beta_{2v}\vee\Box\gamma_{21}\vee\ldots\vee\Box\gamma_{2w})\wedge\ldots\wedge(\alpha_{n1}\vee\ldots\vee\alpha_{nu}\vee\Diamond\beta_{n1}\vee\ldots\vee\Diamond\beta_{nv}\vee\Box\gamma_{n1}\vee\ldots\vee\Box\gamma_{nw})
\wedge  \neg\psi_1\wedge\ldots\wedge\neg\psi_p\wedge\Box\neg\phi_1\wedge\ldots\wedge\Box\neg\phi_q\wedge\Diamond\neg\xi_1\wedge\ldots\wedge\Diamond\neg\xi_r\models\bot$. By distributive law, $(\vee_{\substack{j_{i}=1\\1\leq i\leq n}}^{u}(\wedge_{k=1}^{n}\alpha_{kj_{k}}))\vee (\vee_{\substack{j_{i}=1\\ 1\leq i\leq n}}^{v}(\wedge_{k=1}^{n}\Diamond\beta_{kj_{k}}))\vee (\vee_{\substack{j_{i}=1\\ 1\leq i\leq n}}^{w}(\wedge_{k=1}^{n}\Box\gamma_{kj_{k}}))\\\vee(\vee_{\substack{1\leq j\leq u\\ 1\leq k\leq v}}((\wedge_{\substack{i=1\\ i\neq s}}^{n}\alpha_{ij})\wedge(\wedge_{\substack{s=1\\i\neq s}}^{n}\Diamond\beta_{sk})))\vee (\vee_{\substack{1\leq j\leq u\\ 1\leq k\leq w}}((\wedge_{\substack{i=1\\ i\neq s}}^{n}\alpha_{ij})\wedge(\wedge_{\substack{s=1\\i\neq s}}^{n}\Box\gamma_{sk})))\\\vee (\vee_{\substack{1\leq j\leq v\\ 1\leq k\leq w}}((\wedge_{\substack{i=1\\ i\neq s}}^{n}\Diamond\beta_{ij})\wedge(\wedge_{\substack{s=1\\i\neq s}}^{n}\Box\gamma_{sk})))\vee (\vee_{\substack{1\leq t\leq u\\ 1\leq j\leq v\\ 1\leq k\leq w}}((\wedge_{\substack{m=1\\m\neq i\neq s}}^{n}\alpha_{mt})\wedge(\wedge_{\substack{i=1\\ m\neq i\neq s}}^{n}\Diamond\beta_{ij})\wedge(\wedge_{\substack{s=1\\m\neq i\neq s}}^{n}\Box\gamma_{sk})))\wedge \neg\psi_1\wedge\ldots\wedge\neg\psi_p\wedge\Box\neg\phi_1\wedge\ldots\wedge\Box\neg\phi_q\wedge\Diamond\neg\xi_1\wedge\ldots\wedge\Diamond\neg\xi_r\models\bot$. 

Then by Lemma \ref{several_possibilities}, one of the following will hold.
\begin{enumerate}
\item $(\wedge_{k=1}^{n}\alpha_{kj_{k}})\wedge\neg\psi_1\wedge\ldots\wedge\neg\psi_p\models\bot$ for some $j_{k}$ such that $1\leq j_k\leq u$
\item $(\wedge_{k=1}^{n}\beta_{kj_{k}})\wedge\neg\phi_1\wedge\ldots\wedge\neg\phi_q\models\bot$ for some $j_k$ such that $1\leq j_k\leq v$.
\item $(\wedge_{k=1}^{n}\gamma_{kj_{k}})\wedge\neg\xi_m\wedge\neg\phi_1\wedge\ldots\wedge\neg\phi_q\models\bot$ for some $j_k$ such that $1\leq j_k\leq w$ and for some $m$.
\item $((\wedge_{\substack{i=1\\ i\neq s}}^{n}\alpha_{ij})\wedge(\wedge_{\substack{s=1\\i\neq s}}^{n}\beta_{sk}))\wedge\neg\phi_1\wedge\ldots\wedge\neg\phi_q\models\bot$ for some $j$ and $k$.
\item $((\wedge_{\substack{i=1\\ i\neq s}}^{n}\alpha_{ij})\wedge(\wedge_{\substack{s=1\\i\neq s}}^{n}\gamma_{sk}))\wedge\neg\xi_m\wedge\neg\phi_1\wedge\ldots\wedge\neg\phi_q\models\bot$ for some $j$, $k$ and $m$.
\item $((\wedge_{\substack{i=1\\ i\neq s}}^{n}\beta_{ij})\wedge(\wedge_{\substack{s=1\\i\neq s}}^{n}\gamma_{sk}))\wedge\neg\phi_1\wedge\ldots\wedge\neg\phi_q\models\bot$ for some $j$ and $k$.
\item $((\wedge_{\substack{m=1\\m\neq i\neq s}}^{n}\alpha_{mt})\wedge(\wedge_{\substack{i=1\\ m\neq i\neq s}}^{n}\beta_{ij})\wedge(\wedge_{\substack{s=1\\m\neq i\neq s}}^{n}\gamma_{sk}))\wedge\neg\phi_1\wedge\ldots\wedge\neg\phi_q\models\bot$ for some $t$, $j$ and $k$.
\end{enumerate}

If (1) holds, then $(\wedge_{k=1}^{n}\alpha_{kj_{k}})\models\psi_1\vee\ldots\vee\psi_p$. This implies $(\wedge_{k=1}^{n}\alpha_{kj_{k}})\models C$. As $\alpha_{kj_{k}}\in\overline{L}(U)$ for each $k$, so, $\wedge_{k=1}^{n}\alpha_{kj_{k}}\in\overline{L}(U)$. By assuming $D=\wedge_{k=1}^{n}\alpha_{kj_{k}}$, we have for every implicate $C$ of $X$ there exists $D\in\overline{L}(U)$ such that $D\models C$.  
 
If (2) holds, then $(\wedge_{k=1}^{n}\beta_{kj_{k}})\models\phi_1\vee\ldots\vee\phi_q$. This implies $\Diamond(\wedge_{k=1}^{n}\beta_{kj_{k}})\models\Diamond(\phi_1\vee\ldots\vee\phi_q)$. As $\models\Diamond(p\vee q)\leftrightarrow\Diamond p\vee\Diamond q$, so $\Diamond(\wedge_{k=1}^{n}\beta_{kj_{k}})\models\Diamond\phi_1\vee\ldots\vee\Diamond\phi_q$. Hence, $\Diamond(\wedge_{k=1}^{n}\beta_{kj_{k}})\models C$. By $\Diamond$-$1,2$ rule $\Diamond(\wedge_{k=1}^{n}\beta_{kj_{k}})\in\overline{L}(U)$. By assuming $D=\Diamond(\wedge_{k=1}^{n}\beta_{kj_{k}})$, we have for every implicate $C$ of $X$ there exists $D\in\overline{L}(U)$ such that $D\models C$.  

If (3) holds, then $(\wedge_{k=1}^{n}\gamma_{kj_{k}})\models\xi_m\vee\phi_1\vee\ldots\vee\phi_q$. This implies $\Box(\wedge_{k=1}^{n}\gamma_{kj_{k}})\models\Box(\xi_m\vee\phi_1\vee\ldots\vee\phi_q)$. As $\Box(p\rightarrow q)\models\Diamond p\rightarrow\Diamond q$, so $\Box(\neg p\vee q)\models\neg\Diamond p\vee\Diamond q$. As $\neg\Diamond p\vee\Diamond q\equiv\Box\neg p\vee\Diamond q$ so $\Box(\neg p\vee q)\models\Box\neg p\vee\Diamond q$. So taking $\neg p=p_1$ we get $\Box(p_1\vee q)\models\Box p_1\vee\Diamond q$. So $\Box(\xi_m\vee\phi_1\vee\ldots\vee\phi_q)\models\Box\xi_m\vee\Diamond(\phi_1\vee\ldots\vee\phi_q)\models\Box\xi_m\vee\Diamond\phi_1\vee\ldots\vee\Diamond\phi_q$. Hence  $(\wedge_{k=1}^{n}\Box\gamma_{kj_{k}})\models \Box\xi_m\vee\Diamond\phi_1\vee\ldots\vee\Diamond\phi_q$. So $(\wedge_{k=1}^{n}\Box\gamma_{kj_{k}})\models C$. By $\Box$-rule $\Box\gamma_{kj_{k}}\in\overline{L}(U)$ for each $k$. So $(\wedge_{k=1}^{n}\Box\gamma_{kj_{k}})\in\overline{L}(U)$. By assuming $D=\wedge_{k=1}^{n}\Box\gamma_{kj_{k}}$, we have for every implicate $C$ of $X$ there exists $D\in\overline{L}(U)$ such that $D\models C$.   

If (4) holds, then $((\wedge_{\substack{i=1\\ i\neq s}}^{n}\alpha_{ij})\wedge(\wedge_{\substack{s=1\\i\neq s}}^{n}\beta_{sk}))\models\phi_1\vee\ldots\vee\phi_q$. This implies $\Diamond((\wedge_{\substack{i=1\\ i\neq s}}^{n}\alpha_{ij})\wedge(\wedge_{\substack{s=1\\i\neq s}}^{n}\beta_{sk}))\models\Diamond(\phi_1\vee\ldots\vee\phi_q)$. As $\models\Diamond(p\vee q)\leftrightarrow\Diamond p\vee\Diamond q$, so $\Diamond((\wedge_{\substack{i=1\\ i\neq s}}^{n}\alpha_{ij})\wedge(\wedge_{\substack{s=1\\i\neq s}}^{n}\beta_{sk}))\models\Diamond\phi_1\vee\ldots\vee\Diamond\phi_q$, i.e, $\Diamond((\wedge_{\substack{i=1\\ i\neq s}}^{n}\alpha_{ij})\wedge(\wedge_{\substack{s=1\\i\neq s}}^{n}\beta_{sk}))\models C$. By $\Diamond$-rule, $\Diamond((\wedge_{\substack{i=1\\ i\neq s}}^{n}\alpha_{ij})\wedge(\wedge_{\substack{s=1\\i\neq s}}^{n}\beta_{sk}))\in\overline{L}(U)$. By assuming $D=\Diamond((\wedge_{\substack{i=1\\ i\neq s}}^{n}\alpha_{ij})\wedge(\wedge_{\substack{s=1\\i\neq s}}^{n}\beta_{sk}))$, we have for every implicate $C$ of $X$ there exists $D\in\overline{L}(U)$ such that $D\models C$.    

If (5) holds, then $((\wedge_{\substack{i=1\\ i\neq s}}^{n}\alpha_{ij})\wedge(\wedge_{\substack{s=1\\i\neq s}}^{n}\gamma_{sk}))\models\xi_m\vee\phi_1\vee\ldots\vee\phi_q$. This implies $\Box((\wedge_{\substack{i=1\\ i\neq s}}^{n}\alpha_{ij})\wedge(\wedge_{\substack{s=1\\i\neq s}}^{n}\gamma_{sk}))\models\Box(\xi_m\vee\phi_1\vee\ldots\vee\phi_q)$. In the proof of case (3) we have shown that $\Box(p_1\vee q)\models\Box p_1\vee\Diamond q$. So $\Box((\wedge_{\substack{i=1\\ i\neq s}}^{n}\alpha_{ij})\wedge(\wedge_{\substack{s=1\\i\neq s}}^{n}\gamma_{sk}))\models\Box\xi_m\vee\Diamond(\phi_1\vee\ldots\vee\phi_q)\models\Box\xi_m\vee\Diamond\phi_1\vee\ldots\vee\Diamond\phi_q$. Hence $\Box((\wedge_{\substack{i=1\\ i\neq s}}^{n}\alpha_{ij})\wedge(\wedge_{\substack{s=1\\i\neq s}}^{n}\gamma_{sk}))\models C$. By $\Box$-rule, $\Box((\wedge_{\substack{i=1\\ i\neq s}}^{n}\alpha_{ij})\wedge(\wedge_{\substack{s=1\\i\neq s}}^{n}\gamma_{sk}))\in\overline{L}(U)$. By assuming $D=\Box((\wedge_{\substack{i=1\\ i\neq s}}^{n}\alpha_{ij})\wedge(\wedge_{\substack{s=1\\i\neq s}}^{n}\gamma_{sk}))$, we have for every implicate $C$ of $X$ there exists $D\in\overline{L}(U)$ such that $D\models C$. 

If (6) holds, then $((\wedge_{\substack{i=1\\ i\neq s}}^{n}\beta_{ij})\wedge(\wedge_{\substack{s=1\\i\neq s}}^{n}\gamma_{sk}))\models\phi_1\vee\ldots\vee\phi_q$. This implies $\Diamond((\wedge_{\substack{i=1\\ i\neq s}}^{n}\beta_{ij})\wedge(\wedge_{\substack{s=1\\i\neq s}}^{n}\gamma_{sk}))\models\Diamond(\phi_1\vee\ldots\vee\phi_q)$. As $\models\Diamond(p\vee q)\leftrightarrow\Diamond p\vee\Diamond q$, So $\Diamond((\wedge_{\substack{i=1\\ i\neq s}}^{n}\beta_{ij})\wedge(\wedge_{\substack{s=1\\i\neq s}}^{n}\gamma_{sk}))\models\Diamond\phi_1\vee\ldots\vee\Diamond\phi_q$, i.e, $\Diamond((\wedge_{\substack{i=1\\ i\neq s}}^{n}\beta_{ij})\wedge(\wedge_{\substack{s=1\\i\neq s}}^{n}\gamma_{sk}))\models C$. By $\Diamond$-rule, $\Diamond((\wedge_{\substack{i=1\\ i\neq s}}^{n}\beta_{ij})\wedge(\wedge_{\substack{s=1\\i\neq s}}^{n}\gamma_{sk}))\in\overline{L}(U)$. By assuming $D=\Diamond((\wedge_{\substack{i=1\\ i\neq s}}^{n}\beta_{ij})\wedge(\wedge_{\substack{s=1\\i\neq s}}^{n}\gamma_{sk}))$, we have for every implicate $C$ of $X$ there exists $D\in\overline{L}(U)$ such that $D\models C$.  

If (7) holds, then $((\wedge_{\substack{m=1\\m\neq i\neq s}}^{n}\alpha_{mt})\wedge(\wedge_{\substack{i=1\\ m\neq i\neq s}}^{n}\beta_{ij})\wedge(\wedge_{\substack{s=1\\m\neq i\neq s}}^{n}\gamma_{sk}))\models\phi_1\vee\ldots\vee\phi_q$. So, $\Diamond((\wedge_{\substack{m=1\\m\neq i\neq s}}^{n}\alpha_{mt})\wedge(\wedge_{\substack{i=1\\ m\neq i\neq s}}^{n}\beta_{ij})\wedge(\wedge_{\substack{s=1\\m\neq i\neq s}}^{n}\gamma_{sk})) \models \Diamond(\phi_1\vee\ldots\vee\phi_q)$. As $\models\Diamond(p\vee q)\leftrightarrow\Diamond p\vee\Diamond q$, so  $\Diamond((\wedge_{\substack{m=1\\m\neq i\neq s}}^{n}\alpha_{mt})\wedge(\wedge_{\substack{i=1\\ m\neq i\neq s}}^{n}\beta_{ij})\wedge(\wedge_{\substack{s=1\\m\neq i\neq s}}^{n}\gamma_{sk}))\models\Diamond\phi_1\vee\ldots\vee\Diamond\phi_q$, i.e,  $\Diamond((\wedge_{\substack{m=1\\m\neq i\neq s}}^{n}\alpha_{mt})\wedge(\wedge_{\substack{i=1\\ m\neq i\neq s}}^{n}\beta_{ij})\wedge(\wedge_{\substack{s=1\\m\neq i\neq s}}^{n}\gamma_{sk}))\models C$. By $\Diamond$-rule, $\Diamond((\wedge_{\substack{m=1\\m\neq i\neq s}}^{n}\alpha_{mt})\wedge(\wedge_{\substack{i=1\\ m\neq i\neq s}}^{n}\beta_{ij})\wedge(\wedge_{\substack{s=1\\m\neq i\neq s}}^{n}\gamma_{sk}))\in\overline{L}(U)$. By assuming $D=\Diamond((\wedge_{\substack{m=1\\m\neq i\neq s}}^{n}\alpha_{mt})\wedge(\wedge_{\substack{i=1\\ m\neq i\neq s}}^{n}\beta_{ij})\wedge(\wedge_{\substack{s=1\\m\neq i\neq s}}^{n}\gamma_{sk}))$, we have for every implicate $C$ of $X$ there exists $D\in\overline{L}(U)$ such that $D\models C$.   Hence it is proved. \hfill{$\Box$}
\end{proof}

Now we present the algorithm for computing prime implicates of a modal formula using direct resolution.

\noindent{\bf Algorithm} \hspace{.2cm} { PIC}

\noindent  Input: U, a set of clauses\\
 Output: $\pi(U)$, the set of prime implicates of U\\
{\tt begin\\
$~~~~~$ if $U=\emptyset$\\
$~~~~~~~~~~~$ $\pi(U)=\emptyset$\\
$~~~~~$ else\\
$~~~~~~~~~~~$ $i=1$\\
$~~~~~~~~~~~$ $U_i=U$\\
$~~~~~~~~~~~$ Repeat\\
$~~~~~~~~~~~~~~~$ compute $L(U_i)$\\
$~~~~~~~~~~~~~~~$ $U_{i+1}=\text{Res}(L(U_i))$\\
$~~~~~~~~~~~$ Until $U_{i}=U_{i+1}$\\
$~~~~~~~~~~~$ $\pi(U)=U_{i+1}$\\
$~~~~~$ endif\\
$~~~~~$ return $\pi(U)$\\
end}

\vspace{.2cm}

\noindent{\bf Submodule} \hspace{.2cm} {$L(U^{'})$}

\noindent  Input: $U^{'}$, a set of clauses\\
 Output: $L(U^{'})$, a set of clauses\\
{\tt begin\\
$~~~~~$ if $A\in U^{'}$ and $B\in U^{'}$\\
$~~~~~~~~~~$ then compute $\Theta=\Sigma(A,B)$\\
$~~~~~$ else if $A\in U^{'}$\\
$~~~~~~~~~~$ then compute $\Theta=\Gamma(A)$\\
$~~~~~~$ else\\
$~~~~~~~~~~~$ compute $L(U^{'})=U^{'}\cup\Theta$\\
end}

\vspace{.2cm}

%
%
%
%

\begin{theorem}\label{resolution_closure}{\bf(Soundness and Completeness)}
The set of all prime implicates is a subset of the resolution closure of $U$, i.e, $\pi(U)\subseteq\overline{L}(U)$. Moreover, $\pi(U)=Res(\overline{L}(U))$.
\end{theorem}
 
\begin{proof}
Let $C\in\pi(U)$. This implies $C$ is an implicate of $U$. By Theorem \ref{several_possibilities_wrt_closure}, there exists $D\in\overline{L}(U)$ such that $D\models C$. If $C\not\in \text{Res}(\overline{L}(U))$ then as $D\models C$, so $C\not\in\pi(U)$, which is a contradiction. Hence, $C\in \text{Res}(\overline{L}(U))$. This implies, $\pi(U)\subseteq\text{Res}(\overline{L}(U))$. By Theorem \ref{resolution_consequence}, $\overline{L}(U)\subseteq\Psi(U)$. If $D\in\psi(U)-\overline{L}(U)$, then $D$ is subsumed by some clause $C\in\overline{L}(U)$. Therefore, $\text{Res}(\overline{L}(U))=\text{Res}(\Psi(U))$. As $\pi(U)=\text{Res}(\Psi(U))$, so $\pi(U)=\text{Res}(\overline{L}(U))$. \hfill{$\Box$}
\end{proof}

Now we find the complexity of the above algorithm.

\begin{theorem}
Given a set of clauses $U$, the above algorithm for computing prime implicates requires at most $O(n^{2k})$ resolution and residue operations, where n is the number of clauses in $U$ and k is the maximum number of iterations performed on $U$.
\end{theorem}
\begin{proof}
Let $U_i, 0\leq i\leq k$ be the set of clauses at the end of stage $i$ and $|U|=|U_0|=n$. Let $m_i$ denote the maximum number of clauses in $U_i$ at the end of stage $i$. So $m_0=n$ and $m_1=n+1+\text{min}\{|A|,|B|\}$. This is because when we apply resolution on $A$ and $B$, the maximal recursion depth is at most the number of $\Box, \Diamond, \wedge, \vee$ operators and the number of these operators decreases by 1 with each call and we stop the recursion when it reaches 0. So the number of terminating sub-calls cann't  exceed the total number of $\Box, \Diamond, \wedge, \vee$ operators. So the subcalls will go on till the operators finishes in one of the resolvents $A$ or $B$. That is why we take  $\text{min}\{|A|,|B|\}$. Let $\text{min}\{|A|,|B|\}=l$ for some constant $l$. So $m_1=n+1+l$ and $m_i=m_{i-1}+m_{i-1}m_{i-2}$. So $m_2=m_1+m_1m_0=(n+l+1)+(n+l+1)n$  which is $O(n^2)$. Similarly, $m_3=m_2+m_2m_1$ which is $O(n^3)$. Similarly, $m_k=O(n^k)$. So total number of resolution operations is at most $O(n^k)$, i.e, total number of clauses in $U$ after performing all possible resolution is at most $O(n^k)$. As each clause will check every other clause for residue operation, so the total number of residue operations is $O(\underbrace{n^k+n^k+\ldots+n^k}_{n^k})=O(n^k\times n^k)=O(n^{2k})$. So the total number of resolution and residue operations performed is at most $O(n^{2k})$. \hfill{$\Box$}
\end{proof} 

The above result shows that the algorithm takes exponential time for computing prime implicates and it agrees with the exponential number of prime implicates computed by Chandra \& Markowski \cite{Chandra}.

\begin{theorem}\label{length_submodule}
The length of the formula returned by $L(U^{'})$ at stage $i$ is at most $O(n^{i})$.
\end{theorem}
\begin{proof}
While applying $\Sigma$-rules to an arbitrary pair of clauses $A$ and $B$, the number of propositional variables and the number of binary operators remain same but the number of modal operators decreases by one (for instance $\Box\Diamond-$rule and $\Box\Box-$rule). So the length of $\Theta$ is bounded above by $|A|+|B|-1$. As $m_{i}$ is the maximum number of clauses at stage $i$ and length of each resolved clause is bounded above by $|A|+|B|-1$, so length of $L(U^{'})$ is at most  $m_{i}(|A|+|B|-1)$. As $|A|+|B|-1$ is a constant and $m_i=O(n^{i})$ so length of $L(U^{'})$ is at most $O(n^i)$ at stage $i$.

While applying $\Gamma$-rules to an arbitrary clause $A$ the number of propositional variables, the number of modal operators and the number of binary operators increases but the length of $\Theta$ does not exceed $|A|\times|B|$. As $m_{i}$ is the maximum number of clauses at stage $i$ and length of each resolved clause is bounded above by $|A|\times|B|$ so length of $L(U^{'})$ is at most  $m_{i}(|A|\times|B|)$. As $|A|\times|B|$ is a constant so length of $L(U^{'})$ is at most $O(n^{i})$ at stage $i$.

So from $\Sigma$-rules and $\Gamma$-rules  we conclude that the length of the formula returned by $L(U^{'})$ at stage $i$ is at most $O(n^{i})$.
\end{proof}

\begin{theorem} \label{length_module}
The length of the formulas returned by $PIC$ does not exceed $O(n^{k})$.
\end{theorem}
\begin{proof}
Since the residue operations  are performed for at most $k$ stages and length of $L(U^{'})$ is at most $O(n^{i})$ at stage $i$, so the length of formulas returned by $PIC$ is at most 

$\sum_{i=1}^{k}n^{i}=(n+n^2+\ldots+n^k)=n(1+n+n^2+\ldots+n^{k-1})=n(\frac{n^{k}-1}{n-1})=O(n^{k})$
\end{proof}

\begin{theorem}\label{number_submodule}
The number of formulas returned by $L(U^{'})$ is at most $O(2^{n})$.
\end{theorem}
\begin{proof}
We can note that the number of recursive calls made to $L(U^{'})$ is at most two. Let $M$ be the least number of modal operators applied as a last rule between a pair of clauses. So the total number of recursive calls will be made is at most $M$ as with each call the value of $M$ decreases by $1$ and recursion continues till $M$ becomes $0$. So when $L(U^{'})$ is executed the number of terminating subcalls will not exceed $2^M$. Notice that each call can produce at most one formula and the number of clauses at stage $0$ is $n$, so $M$ is bounded above by $n\times\text{max}\{|A|^{2},|A|+|B|-1\}$ by Theorem \ref{length_submodule}. So the number of formulas output by $L(U^{'})$ is at most $O(2^{n\times\text{max}\{|A|^{2},|A|+|B|-1\}})$, i.e, at most $O(2^{n})$ .
\end{proof}

\begin{theorem}\label{number_module}
The number of prime implicates output by $PIC$ is at most $O(n^{2k}\times 2^{n})$.
\end{theorem}
\begin{proof}
By Theorem \ref{number_submodule}, the number of formulas output by $L(U^{'})$ is at most $O(2^{n})$ and there are $O(n^{2k})$ resolution and residue operations performed in the algorithm, so the number of clauses produced by $PIC$ will be at most $O(n^{2k}\times 2^{n})$. So the number of prime implicates output by $PIC$ is at most $O(n^{2k}\times {2^n})$.
\end{proof}

\begin{theorem}
The length of the smallest clausal representation of a prime implicate of a formula $U$ does not exceed $O(n^{2k+1}\times 2^n)$.
\end{theorem}
\begin{proof}
As the number of prime implicates output by $PIC$ is at most $O(n^{2k}\times 2^n)$ by Theorem \ref{number_module}, so the number of prime implicates of $U$ is at most $O(n^{2k}\times 2^n)$ disjuncts. The length of each disjunct does not exceed $O(n^{i})$ by Theorem \ref{length_submodule}. So the total number of symbols is at most $O(n^{(2k+i)}\times 2^n)$. As $0\leq i\leq k$, so the total number of symbols is at most $O(n^{3k}\times 2^n)$. But there will be $O(n^{2k}\times 2^n-1)$ conjunction symbols which connects $O(n^{2k}\times 2^n)$ disjuncts. So the length of the smallest clausal representation of a prime implicate of a formula is at most $O(n^{3k}\times 2^n+n^{2k}\times 2^n-1)=O((n+1)\times n^{2k}\times 2^n-1)=O((n+1)\times n^{2k}\times 2^n)=O(n^{2k+1}\times 2^n)$.
\end{proof}

\begin{example}
Consider the set of formulas $\{\Diamond(p,\neg p\vee\Box r), \Box\Diamond(\neg r\vee q), \Box\Box(\neg p\vee r)\}$. Now we apply resolution followed by residue operation and we underline those clauses on which resolution takes place. 
\begin{enumerate}
\item Let $U_1=\{\underline{\Diamond(p,\neg p\vee\Box r)}, \Box\Diamond(\neg r\vee q), \Box\Box(\neg p\vee r)\}$.
\item $L(U_1)=\{\Diamond(p,\neg p\vee\Box r), \Box\Diamond(\neg r\vee q), \Box\Box(\neg p\vee r), \Diamond(p,\neg p\vee\Box r,\Box r)\}$. (By (A1), $\vee$ and $\Diamond 1$)
\item $U_2=\{\underline{\Box\Diamond(\neg r\vee q)}, \Box\Box(\neg p\vee r), \underline{\Diamond(p,\neg p\vee\Box r,\Box r)}\}$. 
\item  $L(U_2)=\{\Box\Diamond(\neg r\vee q), \Box\Box(\neg p\vee r), \Diamond(p,\neg p\vee\Box r,\Box r), \Diamond(p,\neg p\vee\Box r,\Box r, \Diamond(\neg r\vee q,q))\}$ (By (A1), $\vee$ and $\Box\Diamond$-rule twice)
\item $U_3=\{\underline{\Box\Diamond(\neg r\vee q)}, \underline{\Box\Box(\neg p\vee r)}, \Diamond(p,\neg p\vee\Box r,\Box r, \Diamond(\neg r\vee q,q))\}$.
\item $L(U_3)=\{\Box\Diamond(\neg r\vee q), \Box\Box(\neg p\vee r), \Diamond(p,\neg p\vee\Box r,\Box r, \Diamond(\neg r\vee q,q)), \Box\Diamond(\neg r\vee q,\neg p\vee q)\}$ (By (A1), $\vee$, $\Box\Box$ and $\Box\Diamond$-rule)
\item $U_4=\{\underline{\Box\Box(\neg p\vee r)}, \Diamond(p,\neg p\vee\Box r,\Box r, \Diamond(\neg r\vee q,q)), \underline{\Box\Diamond(\neg r\vee q,\neg p\vee q)}\}$.
\item $L(U_4)=\{\Box\Box(\neg p\vee r), \Diamond(p,\neg p\vee\Box r,\Box r, \Diamond(\neg r\vee q,q)), \Box\Diamond(\neg r\vee q,\neg p\vee q), \Box\Diamond(\neg r\vee q,\neg p\vee q, \neg p\vee q)\}=\{\Box\Box(\neg p\vee r), \Diamond(p,\neg p\vee\Box r,\Box r, \Diamond(\neg r\vee q,q)), \Box\Diamond(\neg r\vee q,\neg p\vee q), \Box\Diamond(\neg r\vee q,\neg p\vee q)\}$ (By (A1), $\vee$, $\Box\Box$ and $\Box\Diamond$-rule)
\item $U_5= \{\underline{\Box\Box(\neg p\vee r)}, \underline{\Diamond(p,\neg p\vee\Box r,\Box r, \Diamond(\neg r\vee q,q))}, \Box\Diamond(\neg r\vee q,\neg p\vee q)\}$.
\item $L(U_5)= \{\Box\Box(\neg p\vee r), \Diamond(p,\neg p\vee\Box r,\Box r, \Diamond(\neg r\vee q,q)), \Box\Diamond(\neg r\vee q,\neg p\vee q), \Diamond(p,\neg p\vee\Box r,\Box r, \Diamond(\neg r\vee q,q),\Diamond(\neg r\vee q,\neg p\vee q, q))\}$ (By (A1), $\vee$ and $\Box\Diamond$-rule twice)
\item $U_6=  \{\underline{\Box\Box(\neg p\vee r)}, \Box\Diamond(\neg r\vee q,\neg p\vee q), \\\underline{\Diamond(p,\neg p\vee\Box r,\Box r, \Diamond(\neg r\vee q,q),\Diamond(\neg r\vee q,\neg p\vee q, q))}\}$.
\item $L(U_6)=\{\Box\Box(\neg p\vee r), \Box\Diamond(\neg r\vee q,\neg p\vee q), \Diamond(p,\neg p\vee\Box r,\Box r, \Diamond(\neg r\vee q,q),\Diamond(\neg r\vee q,\neg p\vee q, q)), \Diamond(p,\neg p\vee\Box r,\Box r, \Diamond(\neg r\vee q,q),\Diamond(\neg r\vee q,\neg p\vee q, q), \Diamond(\neg r\vee q,\neg p\vee q, q))\}$ (by $\vee$ and $\Box\Diamond$-rule twice).
\item $U_7=\{\underline{\Box\Box(\neg p\vee r)}, \underline{\Box\Diamond(\neg r\vee q,\neg p\vee q)}, \Diamond(p,\neg p\vee\Box r,\Box r, \Diamond(\neg r\vee q,q),\Diamond(\neg r\vee q,\neg p\vee q, q))\}$.
\end{enumerate}
From this step everything gets repeated again and again when we take resolution among these three clauses. So the prime implicates are the the clauses computed in $U_{7}$.
\end{example} 

\section{Conclusion}
In this paper we have suggested an algorithm to compute prime implicates of a modal  formula in $\mathbf{K}$ using resolution method \cite{Enjalbert} and we have also proved its correctness. The algorithm takes exponential time in the size of the original  formula for computing prime implicates and the number of prime implicates are polynomial times exponential times i.e, $O(n^{2k}\times 2^{n})$ in the size of the given formula. So this algorithm is more efficient than the algorithm suggested by Bienvenu \cite{Bienvenu} which computes prime implicates in doubly exponential time. As prime implicates and prime implicants are dual to each other so the proposed algorithm can be used to compute prime implicants of a modal formula. We have also extended the algorithm to compute prime implicates in multi-modal logic using  resolution \cite{Areces} and is yet to be completed and we want to find out its complexity.


%
%
%
%
%
\end{document}